\documentclass[letterpaper,11pt]{article}

\usepackage{amsmath,amsfonts,amssymb,amsthm}
\usepackage{mathbbol}
\usepackage{amsthm}
\usepackage{graphicx,color}
\usepackage{boxedminipage}
\usepackage[linesnumbered, boxed, algosection,noline,noend]{algorithm2e}
\usepackage{framed}
\usepackage{thmtools}
\usepackage{thm-restate}
\usepackage{xspace}
\usepackage{todonotes}
\usepackage{vmargin}
\setmarginsrb{2cm}{2cm}{2cm}{2cm}{0pt}{0pt}{0pt}{6mm}
 \usepackage[pdftex, plainpages = false, pdfpagelabels, 
                 bookmarks=false,
                 bookmarksopen = true,
                 bookmarksnumbered = true,
                 breaklinks = true,
                 linktocpage,
                 pagebackref,
                 colorlinks = true,  
                 linkcolor = blue,
                 urlcolor  = blue,
                 citecolor = red,
                 anchorcolor = green,
                 hyperindex = true,
                 hyperfigures
                 ]{hyperref} 
\usepackage{xifthen}
\usepackage{tabularx}
\usepackage{capt-of}
\usepackage{thmtools}
\usepackage{thm-restate}
\usepackage{subcaption}
\usepackage{booktabs}
\usepackage{calc}
\usepackage{dirtytalk}
\usepackage{bm}
\usepackage{enumitem}

\usetikzlibrary{calc}
\usetikzlibrary{shapes}
\usetikzlibrary{arrows.meta}

\theoremstyle{plain}
\declaretheorem[name=Theorem]{theorem}
\newtheorem{lemma}[theorem]{Lemma}

\theoremstyle{definition}
\newtheorem{definition}[theorem]{Definition}

\newtheorem{observation}[theorem]{Observation}
\theoremstyle{remark}

\newtheorem{claim}{Claim}

\newcommand{\np}{\mathsf{NP}}
\newcommand{\rp}{\mathsf{RP}}
\newcommand{\I}{\mathcal{I}}
\newcommand{\C}{\mathcal{C}}
\newcommand{\A}{\mathcal{A}}
\newcommand{\cs}{\mathscr{S}}\usepackage{mathrsfs}
\newcommand{\cl}{{\mathscr{L}}}
\newcommand{\Op}{\ensuremath{\cal O}}
\newcommand{\cb}{\mathscr{B}}



\newlength{\RoundedBoxWidth}
\newsavebox{\GrayRoundedBox}
\newenvironment{GrayBox}[1]%
   {\setlength{\RoundedBoxWidth}{.93\textwidth}
    \def\boxheading{#1}
    \begin{lrbox}{\GrayRoundedBox}
       \begin{minipage}{\RoundedBoxWidth}}%
   {   \end{minipage}
    \end{lrbox}
    \begin{center}
    \begin{tikzpicture}%
       \node(Text)[draw=black!20,fill=white,rounded corners,%
             inner sep=2ex,text width=\RoundedBoxWidth]%
             {\usebox{\GrayRoundedBox}};
        \coordinate(x) at (current bounding box.north west);
        \node [draw=white,rectangle,inner sep=3pt,anchor=north west,fill=white] 
        at ($(x)+(6pt,.75em)$) {\boxheading};
    \end{tikzpicture}
    \end{center}}     

\newenvironment{defproblemx}[2][]{\noindent\ignorespaces%
                                \FrameSep=6pt%
                                \parindent=0pt%
                \vspace*{-1.5em}
                \ifthenelse{\isempty{#1}}{%
                  \begin{GrayBox}{\textsc{#2}}%
                }{%
                  \begin{GrayBox}{\textsc{#2} parameterized by~{#1}}%
                }
                \begin{tabular*}{\textwidth}{@{\hspace{.1em}} >{\itshape} p{1.8cm} p{0.8\textwidth} @{}}%
            }{
                \end{tabular*}%
                \end{GrayBox}%
                \ignorespacesafterend
            }  

\usepackage{authblk}
\author[1]{Sayan Bandyapadhyay\thanks{sayan.bandyapadhyay@gmail.com}}
\affil[1]{Department of Informatics, University of Bergen, Norway}
\date{}

\begin{document}
\title{On Perturbation Resilience of Non-Uniform $k$-Center} 
\maketitle

\begin{abstract}
The Non-Uniform $k$-center (NUkC) problem has recently been formulated by Chakrabarty, Goyal and Krishnaswamy [ICALP, 2016] as a generalization of the classical $k$-center clustering problem. In NUkC, given a set of $n$ points $P$ in a metric space and non-negative numbers $r_1, r_2, \ldots , r_k$, the goal is to find the minimum dilation $\alpha$ and to choose $k$ balls centered at the points of $P$ with radius $\alpha\cdot r_i$ for $1\le i\le k$, such that all points of $P$ are contained in the union of the chosen balls. They showed that the problem is $\np$-hard to approximate within any factor even in tree metrics. On the other hand, they designed a ``bi-criteria'' constant approximation algorithm that uses a constant times $k$ balls. Surprisingly, no true approximation is known even in the special case when the $r_i$'s belong to a fixed set of size 3. In this paper, we study the NUkC problem under perturbation resilience, which was introduced by Bilu and Linial [Combinatorics, Probability and Computing, 2012]. We show that the problem under 2-perturbation resilience is polynomial time solvable when the $r_i$'s belong to a constant sized set. However, we show that perturbation resilience does not help in the general case. In particular, our findings imply that even with perturbation resilience one cannot hope to find any ``good'' approximation for the problem. 
\end{abstract}

\section{Introduction}
\textit{Stability} is a popular notion, which has been used in literature in the context of \textit{beyond worst case analysis}. The general idea is to impose extra constraints on the inputs such that the (stable) instances that satisfy those constraints can capture the instances that appear in real life applications. In other words, we would like to exclude the ``unrealistic'' instances from consideration and obtain optimistic bounds for algorithms on the remaining inputs. For example, a major collection of work along this line have focused on designing polynomial time algorithms for $\np$-complete problems under different stability conditions. Bilu and Linial \cite{BiluL12} introduced a notion of stability, which they termed as $\psi$-\textit{perturbation resilience} for some $\psi > 1$. Informally, an instance is called $\psi$-perturbation-resilient if the optimal solution remains same even after the instance is perturbed by a factor of $\psi$.   

%
Recently, researchers have shown sufficient interest in studying geometric clustering problems under perturbation resilience. An instance of a clustering problem is  $\psi$-perturbation-resilient if the optimal clustering is unique and remains unchanged under $\psi$-factor perturbation of the input distances. Awasthi {\em et al.} \cite{AwasthiBS12} showed that the standard center based clustering problems (e.g. $k$-center, $k$-median) can be solved in polynomial time under $\psi$-perturbation-resilience for $\psi\ge 3$. In any such center based clustering problem, the clustering is obtained by assigning a point to its nearest center. In other words, the clustering is induced by the Voronoi partition of the points w.r.t. the chosen centers. Subsequently, Balcan and Liang \cite{BalcanL16} designed a polynomial time algorithm for these clustering problems under $\psi$-perturbation-resilience for $\psi\ge 1+\sqrt{2}$, improving the bound of Awasthi {\em et al.} \cite{AwasthiBS12}. Later, Balcan {\em et al.} \cite{BalcanHW16} improved the bound for $k$-center to 2. On the other hand, they showed that $k$-center under $\psi$-perturbation-resilience cannot be solved in polynomial time for $\psi < 2$, unless $\np=\rp$. They also considered the more general asymmetric $k$-center problem, where the distances are not necessarily symmetric (but satisfy triangle inequality). The problem is known to not admit a constant approximation unless $\np \subseteq \textsf{DTIME}(n^{\log \log n})$, where $n$ is the input size \cite{ChuzhoyGHKKKN05}. Surprisingly, Balcan {\em et al.} \cite{BalcanHW16} showed that asymmetric $k$-center under $2$-perturbation-resilience can be solved in polynomial time. Angelidakis {\em et al.} \cite{AngelidakisMM17} gave a generic polynomial time algorithm for clustering problems with center based objectives (e.g. $k$-center, $k$-median, $k$-means) under $2$-perturbation-resilience. Recently, Cohen-Addad and Schwiegelshohn \cite{Cohen-AddadS17} proved that a simple local search scheme yields optimal solutions for problems like $k$-median and $k$-means, under $\psi$-perturbation-resilience for $\psi > 3$. Chekuri and Gupta \cite{ChekuriG18} showed that an LP relaxation of $k$-center under $2$-perturbation-resilience admits an integral solution. They also proved the same result for $k$-center with outliers. Balcan and Liang \cite{BalcanL16} introduced a weaker stability assumption called $(\psi,\epsilon)$-perturbation-resilience, where the optimal solution under $\psi$-perturbation can differ in at most $\epsilon$ fraction of the points from the original optimal clustering (see Preliminaries for the formal definition). Assuming that each cluster contains more than $2\epsilon n$ points, Balcan {\em et al.} \cite{BalcanHW16} showed that $k$-center under $(3,\epsilon)$-perturbation-resilience can be solved in polynomial time, where $n$ is the number of input points. 

The increasing interest in studying perturbation resilient clustering has given rise to several open directions. One such interesting direction is to study clustering problems, where the clustering is not necessarily induced by Voronoi partition. One such clustering problem is \textit{Non-Uniform $k$-center} (NUkC). In NUkC, we are given a set of $n$ points $P$ in a metric space, non-negative integers $r_1, r_2, \ldots , r_k$, and the goal is to find the minimum dilation $\alpha$ and to choose $k$ balls centered at the points of $P$ with radius $\alpha\cdot r_i$ for $1\le i\le k$, such that all points of $P$ are contained in the union of the chosen balls. We refer to any feasible solution of this problem composed of the chosen balls as a \emph{feasible placement}. From a feasible placement, a clustering is retrieved in the following way -- each point is assigned to a fixed ball that contains the point, and then for each ball, the points that are assigned to that ball form a cluster. Figure \ref{fig:nonvoronoi} shows that, the optimal clustering for an instance of NUkC is not the same as the Voronoi partition w.r.t. the centers of the balls in the optimal placement. The NUkC problem was formulated by Chakrabarty~{\em et al.}~\cite{ChakrabartyGK16} as a generalization of the well-studied $k$-center clustering problem, where all $r_i$'s are same. Apart from clustering, NUkC has several applications in vehicle routing, sensor placement, and so on. For example, in vehicle routing, we need to find $k$ depot locations corresponding to $k$ vehicles having different speeds, such that any customer can be served by some vehicle as quickly as possible. 

\begin{figure}[t]
\centering
\includegraphics[width=.25\linewidth]{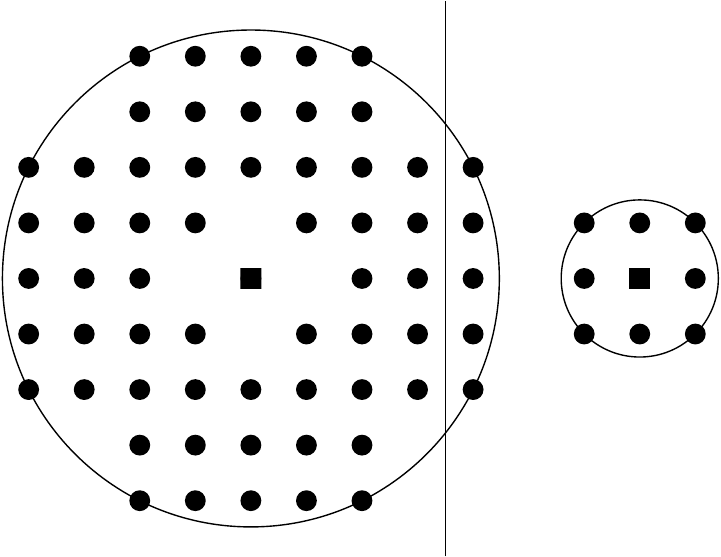}
\caption{The optimal clusters are contained in the two disks in the Euclidean plane. The centers of balls are shown by boxes. The Voronoi partition w.r.t. the centers contains two subsets of points lying on the different sides of the vertical bisector line.}
\label{fig:nonvoronoi}
\end{figure}

As mentioned before, $k$-center is a special case of NUkC where all the input radii are equal. We call this version of the problem as NUkC with one radius class. In general, all the radii might not be equal. But, we can consider only distinct radii from the input and associate a multiplicity parameter $k_i$, with each such radius $r_i$, which denotes the number of balls of radius $r_i$ that can be opened. 
Then the problem can be formulated equivalently in the following way. 
\begin{definition}
 (NUkC with $t$ radii classes) Given a set of $n$ points $P$ in a metric space, $t \le k$ distinct radii $r_1 > r_2 > \ldots > r_t$ and non-negative integers $k_1,\ldots,k_t$ such that $\sum_{i=1}^t k_i=k$, the goal is to find the minimum dilation $\alpha$ and to choose $k_i$ balls centered at the points of $P$ with radius $\alpha\cdot r_i$ for all $1\le i\le t$, such that the union of the chosen balls contains all the input points. 
\end{definition}

We note that $k$-center with outliers is a special case of NUkC with 2 radii classes where the radius $r_2=0$. Using a reduction from the Firefighters problem \cite{AdjiashviliBZ17}, Chakrabarty~{\em et al.} (Theorem 2 in \cite{ChakrabartyGK16}) proved that NUkC is $\np$-hard to approximate within any constant factor even in tree metrics. In fact, their construction proves $c$-inapproximability of the problem for any $c$, not necessarily a constant. On the other hand, they designed a $(c_1,c_2)$ bi-criteria approximation for the problem for large constants $c_1$ and $c_2$, i.e., if the algorithm is allowed to use $c_1\cdot k_i$ balls of type $i$ (thus $c_1\cdot k$ in total), it can produce a solution with dilation at most $c_2$ times the optimal dilation. 
They also gave a $1+\sqrt{5}$-approximation for NUkC with two radii classes. For $k$-center with outliers, they gave an improved 2-approximation. However, even when the number of distinct radii is 3, no true approximation is known.


The motivation behind the study of NUkC under perturbation resilience is that, in many 
applications, the distance function is heuristic. In fact, when the points represent
structures like images, proteins, documents, etc., it is very hard to find the true distance function, and various standard
distance/dissimilarity measures are used. If one solves a clustering problem with such a heuristic
distance function and expects good results, then they implicitly assume that the optimal solution
of the problem is not sensitive to small perturbations of the distance function. The perturbation resilience
condition is a natural way to
make this implicit assumption precise. And, the separation between the clusters forces an optimal clustering to be unique. 

\paragraph{Our results.} In this paper, we obtain the following results. 

\begin{enumerate}
 \item \textit{Polynomial time exact algorithm for NUkC with a constant number of radii classes under ``$2$-perturbation-resilience'' and ``$(3,\epsilon)$-perturbation-resilience when each cluster contains more than $\epsilon n+1$ points''.} Our algorithm reduces the NUkC problem to a version of Firefighters problem on trees (formally defined in Section \ref{sec:constantclasses}). Under the stability assumptions, we can show that a feasible solution of NUkC maps to a feasible solution of Firefighters problem and vice versa. Here, in particular, we use the ``well-separated'' structure of the clusters in the optimal clustering that follows due to stability. The reduction has the property that if NUkC has $t$ distinct radii classes, then the height of the constructed tree instance is $t+1$. Then we show that using a dynamic programming based scheme the Firefighters problem can be solved in polynomial time for constant height tree instances. Thus we also obtain a polynomial time algorithm for NUkC under perturbation resilience with a constant number of radii classes. We note that the algorithms for center based clustering problems in \cite{AngelidakisMM17,BalcanL16,ChekuriG18} are also based on tree computation and dynamic programming. However, the structure of the tree we compute is very different.
 We also note that our result under $2$-perturbation-resilience is tight, as even for $k$-center it is unlikely to obtain a polynomial time algorithm under $\psi$-perturbation-resilience for $\psi < 2$. To prove the result for $(3,\epsilon)$-perturbation-resilience, we assume that each cluster contains more than $\epsilon n+1$ points. We note that such a lower bound is necessary, as in its absence even $k$-center is $\np$-hard \cite{BalcanHW16} under $(\psi,\epsilon)$-perturbation-resilience for all $\psi \ge 1$ and $\epsilon > 0$. 
 
 \item \textit{$\gamma$-inapproximability for NUkC under ${\psi}$-perturbation-resilience for any $\gamma > 1$ and $\psi\le \gamma$,
  unless $\np=\rp$.} 
 Our result implies that, for any $\psi,\gamma > 1$, even with $\psi$-perturbation-resilience one cannot hope to find a $\gamma$-approximation for the problem. Our result should be contrasted with the polynomial time algorithm for asymmetric $k$-center under $2$-perturbation-resilience, as asymmetric $k$-center is another candidate problem which is hard to approximate within a constant factor. To prove the result, we use a chain of reductions starting from the satisfiability problem to the NUkC problem in tree metrics under perturbation resilience assumption. 
The last reduction in the chain is from a version of the Firefighters problem which shows that NUkC is hard to approximate within a factor of $\gamma$ in tree metrics for any $\gamma$. 
Our reduction is similar to the reduction in \cite{ChakrabartyGK16}. Then, we argue that the constructed tree instances of NUkC are ${\gamma}$-perturbation-resilient, and hence the similar hardness follows even for NUkC under ${\gamma}$-perturbation-resilience. We also extend this hardness result to Euclidean metric of dimension $d$ for $d \ge 1$ using a classical tree embedding result of Gupta \cite{Gupta00}.   
 
\end{enumerate}

The main contribution of this paper is twofold. The first one is to be able to establish an exact connection between NUkC under perturbation resilience and the Firefighters problem on trees. To establish this connection, we need to prove that perturbation resilience implies that the optimal clusters are ``well-separated''. Similar properties have been proved in the context of other problems (e.g., $k$-center). Our contribution is to be able to extend these proofs for NUkC as well. However, the extension is non-trivial, and one need sufficiently good amount of extra work, as here we need to deal with non-uniform radii. We note that Chakrabarty~{\em et al.} \cite{ChakrabartyGK16} also showed a reduction from NUkC to Firefighters. However, their LP-aware reduction is very different. Our second contribution is the tight hardness result for the problem. This result along with the polynomial time algorithm gives the complete picture for NUkC under perturbation resilience. To prove this result we are faced with the following challenges. In any such hardness construction, one needs to show that the instances of NUkC to which we map are perturbation resilient. Thus, we need to show that these instances have unique optimal solution and the optimal solution does not change with some perturbation of the distances. Chakrabarty~{\em et al.} \cite{ChakrabartyGK16} showed a reduction from Firefighers to NUkC. However, using their distance function it is not straightforward to show that the constructed instances are insensitive to the perturbation of distances. Nevertheless, we consider a similar distance function and show the reduction works out well with this modification. To prove the uniqueness of the optimal solutions, we reduce a ``unique'' version of 3SAT to a ``unique'' version of Firefighters using a chain of reductions.  

\paragraph{Related work and Open questions.} Other optimization problems have also been studied under stability assumptions \cite{AngelidakisMM17,deshpande19a,FriggstadKS19,MakarychevMV14,MihalakSSW11}. Also different stability assumptions have been introduced and well-studied in the literature \cite{AwasthiBS10,KumarK10,OstrovskyRSS12}. Most of the clustering problems mentioned here are $\np$-hard, but admit some constant approximations, e.g., see \cite{AhmadianNSW17,ByrkaPRST17,Gonzalez85,HochbaumS85} and the references therein. It would be interesting to see if one can obtain a constant approximation for NUkC with a constant number of radii classes without any perturbation resilience assumptions. Also, one can study similar hard clustering problems (e.g., $k$-clustering \cite{BandyapadhyayV16}) under perturbation resilience. 

\paragraph{Organization.} In Section \ref{sec:prelims}, we define some notations that we use throughout the paper, and make a few observations that will be useful later. In Section \ref{sec:properties}, we list some properties implied due to perturbation resilience of the input instances. Then in Section \ref{sec:constantclasses}, we discuss the algorithm for NUkC with any constant number of classes and prove its correctness by using the properties proved in the previous section. Lastly, in Section \ref{sec:hardness}, we prove the hardness results for the general problem. The proofs of lemmas and theorems
marked with ($*$) appear in the Appendix.


\section{Preliminaries}\label{sec:prelims}
We denote an instance of NUkC with $t$ radii classes on metric $d$ by $(P,d,t)$. Note that the radii ($r_i$) and multiplicity ($k_j$) parameters remain implicit in this notation. But, references to these parameters will become clear from the context. A ball with center $p\in P$ and radius $r$, denoted by $B(p,r)$, is the set of points $\{q\in P\mid d(p,q)\le r\}$. A set of balls \textit{covers} a set of points if the union of the balls contains all the points. Recall that a feasible placement is a feasible solution of the problem composed of the chosen balls that cover all the input points. A feasible NUkC \textit{clustering} $\C$ of the input set of points $P$ is a partition $\{C_1,\ldots,C_k\}$, such that there is a feasible placement $\Pi$ with the property that for all $i$, $C_i$ is a subset of a ball in the placement. We say that the clustering $\C$ is \textit{induced} by the placement $\Pi$. The radius of a cluster $C$ w.r.t. any distance function $d$, denoted by c-radius$(C,d)$, is $\min_{p\in P} \max_{q\in C} d(p,q)$. Note that no ball centered at a point $p\in P$ of radius smaller than c-radius$(C,d)$ can cover all the points of $C$. 
For a placement with dilation $\alpha$, a ball with radius $\alpha r_i$ (resp. $<\alpha r_i$ and $\ge \alpha r_i$) is called an $r_i$ (resp. $<r_i$ and $\ge r_i$) -ball. 

Consider a metric space $P$ with metric $d: P\times P \rightarrow {\mathbb{R}}_{\ge 0}$. A metric $d_1$ is called a $\psi$-perturbation of $d$ if for any $p,q \in P$, $d(p,q)/\psi \le  d_1(p,q) \le d(p,q)$\footnote{One can also define $\psi$-perturbation by both increasing and decreasing the distances - the two definitions are equivalent modulo some factor, as one can always scale the input distances appropriately.}. 
In this paper, all perturbations we consider satisfy the metric properties. 

\begin{definition}
 An instance $\mathcal{I}=(P,d,t)$ of NUkC is called $\psi$-perturbation-resilient ($\psi$-PR) if for any metric $\psi$-perturbation $d_1$ of $d$, the unique optimal NUkC clustering of $\mathcal{I}'=(P,d_1,t)$ is identical to the unique optimal clustering of $\mathcal{I}$.
\end{definition}

Note that in general, optimal clustering of NUkC might not be unique. We refer to the instance $\I'$ as a $\psi$-perturbed instance of $\I$. A few examples demonstrating the definition of perturbation resilience w.r.t. NUkC with $t$ radii classes are shown in Appendix \ref{ap:examples}. 
We also consider another notion of perturbation resilience introduced by Balcan and Liang \cite{BalcanL16}, where the optimal clustering is allowed to be different by a few points when the distances are perturbed. Here we rewrite this notion in terms of NUkC. Two clusterings $\C=\{C_1,\ldots,C_k\}$ and $\C'=\{C_1',\ldots,C_k'\}$ are called $\epsilon$-close if at most $\epsilon n$ points are clustered differently in the two clusterings, i.e., the minimum value of $\sum_{i=1}^k |C_i\setminus C_{f(i)}'|$ over all permutations $f$ of $\{1,2,\ldots,k\}$ is at most $\epsilon n$. 

\begin{definition}
 An instance $\mathcal{I}=(P,d,t)$ of NUkC is called $(\psi,\epsilon)$-perturbation-resilient ($(\psi,\epsilon)$-PR) if for any metric $\psi$-perturbation $d_1$ of $d$, any optimal NUkC clustering of $\mathcal{I}'=(P,d_1,t)$ is $\epsilon$-close to any optimal clustering of $\mathcal{I}$.
\end{definition}


This is again a well-studied stability criterion \cite{AgarwalJP15}. Note that when $\epsilon=0$, any optimal NUkC clustering of $\mathcal{I}'$ must be same as any optimal clustering of $\mathcal{I}$. This implies that optimal clustering of $\mathcal{I}$ and $\mathcal{I}'$ are unique and we obtain the definition of $\psi$-PR. Thus, if an instance of NUkC is $\psi$-PR, then it is also $(\psi,0)$-PR, and hence any hardness result for NUkC under $\psi$-PR trivially follows for NUkC under $(\psi,\epsilon)$-PR. Now, we have the following simple observation, which will be useful later in proving the properties of the PR instances. 
%
%

\begin{observation}
[$*$]
\label{obs:uniquecluster}
 Consider an NUkC instance $\mathcal{I}=(P,d,t)$ that admits a unique optimal clustering $\Op$. Let $C$ be any cluster in $\Op$. Also, consider an optimal placement $\Pi$ where $C$ is covered by a ball $B$. Then, the following two properties hold.
 \begin{itemize}
  \item The center $p$ of the ball $B$ must belong to $C$. 
  \item For any two points $u,v$ that lie in two different clusters of $\Op$, both of $u,v$ cannot be contained in $B$.
 \end{itemize}
\end{observation}

WLOG we can assume that the optimal dilation of a $\psi$-PR or a $(\psi,\epsilon)$-PR instance of NUkC is 1. Like in the general case without perturbation resilience, in this case also the assumption can be introduced by scaling $r_i$ values by a guessed value of the optimal dilation $\alpha$. 

%
\begin{lemma}
[$*$]
\label{lem:dila1}
 Suppose there is a polynomial time algorithm $\A$ for the NUkC problem with $t$ radii classes under $\psi$-PR (resp. $(\psi,\epsilon)$-PR) with the properties that (i) for an instance which admits a feasible placement of balls with dilation 1, $\A$ returns ``yes'' and a feasible clustering, 
 and (ii) for an instance which does not admit a feasible placement of balls with dilation 1, $\A$ returns ``no''. Then, the NUkC problem with $t$ radii classes under $\psi$-PR (resp. $(\psi,\epsilon)$-PR) can be solved in polynomial time.
\end{lemma}

\section{Properties of Perturbation Resilience}\label{sec:properties}

In this section, we show that perturbation resilience imposes useful structure on the optimal solution. 
First, we consider the instances under $(\psi,\epsilon)$-perturbation resilience with $\psi=3$ and prove an interesting property of the optimal clustering. 

\begin{lemma}\label{lem:well-separatedrobust}
 Consider any optimal placement $\Pi$ for a $(3,\epsilon)$-PR NUkC instance ${\I}=(P,d,t)$ with optimal dilation 1 where the size of each optimal cluster is $> \epsilon n+1$. Let $C_1$ and $C_2$ be two clusters induced by two balls of $\Pi$ with radii $r_i$ and $r_j$, respectively with $r_i \ge r_j$. Then, for any $p\in C_1$ and $q\in C_2$, $d(p,q) > r_i$. 
\end{lemma}

\begin{proof}
 Let $\Op$ be an optimal clustering of $\I$ that is induced by $\Pi$ and contains $C_1, C_2$ as clusters. For the sake of contradiction, suppose there are two points $p\in C_1$ and $q\in C_2$ such that $d(p,q) \le r_i$. Then, we show that there is a $3$-perturbation $d'$ of $d$ such that an optimal clustering of $\I'=(P,d',t)$ is not $\epsilon$-close to $\Op$. But, this gives a contradiction to the assumption that $\I$ is a $(3,\epsilon)$-PR instance, and hence the lemma follows. 
 
 To construct the $3$-perturbation $d'$ of $d$, we at first construct another metric $d_1$. Later we will scale $d_1$ to construct $d'$. Let $B_1=B(c_1,r_i)$ and $B_2=B(c_2,r_j)$ be the balls in $\Pi$ that induce $C_1$ and $C_2$, respectively. Then, for any $s\in C_2$, $d(p,s) \le d(p,q)+d(q,s)\le r_i+2r_j\le 3r_i$. Also, for any $w\in C_1$, $d(p,w) \le 2r_i$. 
 First, we construct a complete graph $G$ with vertex set equal to $P$, and for any edge $(u,v)$, its length is defined by the function $l$ as follows. 
 
 \[l(u,v) =
\left\{
	\begin{array}{ll}
		3r_i  & \mbox{if } u=p, v \in (C_1\cup C_2)\setminus \{c_1\}  \text{ and } d(u,v)\ge r_i\\
		3\cdot d(u,v) & \mbox{otherwise} 
	\end{array}
\right.\]

The distance $d_1$ is the shortest path metric on $G$. Note that, as mentioned before, for any $v\in (C_1\cup C_2)\setminus \{c_1\}$, $d(p,v) \le 3r_i$. Thus, it is not hard to see that, for any $u,v \in P$, $d(u,v)\le d_1(u,v) \le 3\cdot d(u,v)$. Now, let us define the metric $d'$. For any two points $u,v$, $d'(u,v)=d_1(u,v)/3$. Hence, for any $u,v \in P$, $d(u,v)/3\le d'(u,v) \le d(u,v)$. It follows that $d'$ is a metric $3$-perturbation of $d$, and thus the optimal clustering of $\I'$ is $\epsilon$-close to $\Op$. 

Now, let $\I_1=(P,d_1,t)$ and ${\Op}_1$ be an optimal clustering of $\I_1$. As scaling does not change optimality of a clustering (for a formal proof see the proof of Lemma \ref{lem:dila1}), ${\Op}_1$ is also an optimal clustering of the instance $\I'=(P,d',t)$. Thus ${\Op}_1$ is $\epsilon$-close to $\Op$. 
Next, we prove the following claim. 

\begin{claim}\label{cl:opt3}
 The optimal dilation of $\I_1$ is 3.
\end{claim}

\begin{proof}
  As for any $u,v \in V$, $d_1(u,v)\le 3\cdot d(u,v)$, the optimal dilation of $\I_1$ is at most 3. We prove that this dilation is at least 3. Suppose the dilation is less than 3. 
 Let $\Pi'$ be any placement with dilation less than 3 that induces the optimal clustering ${\Op}_1$ of $\I_1$. Then, we show that ${\Op}_1$ is also a feasible clustering of $\I$ with dilation less than 1. But, this is a contradiction, and hence the claim follows. Next, given $\Pi'$, we show the existence of a placement for $\I$ with dilation less than 1 that induces ${\Op}_1$. 
 
 Consider any cluster $C'\in {\Op}_1$, and suppose it gets covered by an $r_l$-ball $B=B(w,r)$ in $\Pi'$. Let $x$ be any point in $C'$. Now, consider the distance $d_1$. Let $\pi$ be any shortest path between $w$ and $x$. We claim that $\pi$ cannot contain the edge $(p,v)$ for any $v\in (C_1\cup C_2)\setminus \{c_1\}$ with $d(p,v)\ge r_i$. For the sake of contradiction, say $\pi$ contains $(p,v)$. 
 Note that $d_1(p,v)=3r_i$. As $\pi$ contains $(p,v)$, $d_1(w,p) \le r-3r_i$. Now, consider any point $u\in (C_1\cup C_2)\setminus \{c_1\}$. If $d(p,u) \ge r_i$, $d_1(p,u)=3r_i$. Otherwise, $d(p,u) < r_i$, and thus $d_1(p,u)=3\cdot d(p,u) < 3r_i$. Thus, $d_1(w,u)\le d_1(w,p)+d_1(p,u)\le r$. Hence, all the points of $(C_1\cup C_2)\setminus \{c_1\}$ are in $B$. But, as $C_1,C_2$ contain more than $\epsilon n +1$ points, it follows that there is an optimal clustering of $\I_1$ that is not $\epsilon$-close to $\Op$. Thus, we get a contradiction. Hence, $\pi$ does not contain $(p,v)$, and thus from the definition of the metric $d_1$, it follows that $d_1(w,x)= 3\cdot d(w,x)$. Thus, a ball centered at $w$ and having radius $r/3$ can cover the points of $C'$ in $\I$. Now, note that $r < 3r_l$, and thus $r/3 < r_l$. Hence, it is sufficient to use an $r_l$-ball with less than $ 1$ factor expansion to cover the points of $C'$ in $\I$. In our new placement for $\I$, we use the $r_l$-ball $B(w,r/3)$ corresponding to each such cluster $C'$.  
Clearly, the dilation of the new placement is less than 1. 
\end{proof}

Now, we show a clustering ${\Op}_2$ of $\I_1$ that contains exactly $k$ clusters, has dilation 3 and is not $\epsilon$-close to $\Op$. ${\Op}_2$ contains all the clusters in $\Op$ except $C_1$ and $C_2$, and the clusters $(C_1\cup C_2)\setminus \{c_1\}, \{c_1\}$. Note that for any $s\in (C_1\cup C_2)\setminus \{c_1\}$, $d(p,s) \le 3r_i$. Thus, $(C_1\cup C_2)\setminus \{c_1\}$ can be covered by a ball of radius $3r_i$. It follows that the dilation of ${\Op}_2$ is at most 3 and hence it is an optimal clustering. Clearly, the two clusterings $\Op$ and ${\Op}_2$ differ in $> \epsilon n$ points, as $|C_1|> \epsilon n+1$ and $|C_2|> \epsilon n+1$. Now, for the same reason mentioned before, ${\Op}_2$ is also an optimal clustering of the instance $\I'=(P,d',t)$. Hence, $d'$ is the desired $3$-perturbation. This completes the proof of the lemma. 
\end{proof}

In the proof of the above lemma, one could have defined $d'$ directly without going via $d_1$. However, for simplicity of exposition, we have followed this approach. Indeed, this approach shows that if one defines $\psi$-perturbation by increasing the (instead of decreasing) distances, the lemma still holds. A proof can directly use the 3-perturbation $d_1$ in that case.

Note that, as a 3-PR instance is also a $(3,0)$-PR instance, the above lemma trivially follows for 3-PR instances. In the following, we will show that the above mentioned property of the optimal clustering follows even for any 2-PR instance.

\begin{lemma}
[$*$]
\label{lem:well-separated}
 Consider any optimal placement $\Pi$ for a $2$-PR NUkC instance $\mathcal{I}=(P,d,t)$ with optimal dilation 1. Let $C_1$ and $C_2$ be two clusters induced by two balls of $\Pi$ with radius $r_i$ and $r_j$, respectively, where $r_i \ge r_j$. Then, for any $p\in C_1$ and $q\in C_2$, $d(p,q) > r_i$. 
\end{lemma}

\section{NUkC with a Constant Number of Radii Classes}\label{sec:constantclasses}
In this section, we show a polynomial time reduction from NUkC to the Constrained Resource Minimization for Fire Containment on Trees problem.

\begin{definition}
 (Constrained Resource Minimization for Fire Containment on Trees (CRMFC-T)). Given a rooted tree $T=(V,E)$ with height $t+1$, a set of forbidden nodes $F\subseteq V$, and integers $k_1,\ldots,k_t$, the goal is to decide if there is a collection of non-root nodes $U\subseteq (V\setminus F)$ such that (a) for every leaf-root path $\pi$, $U$ contains at least one node from $\pi$, and (b) $|U\cap L_i|\le k_i$ for $1\le i\le t$, where $L_i$ is the layer $i$ nodes of $T$, i.e., the nodes at distance exactly $i$ from the root. 
\end{definition}

Given any instance $\I=(P,d,t)$ of NUkC under 2-PR or $(3,\epsilon)$-PR (the size of each optimal cluster is more than $\epsilon n+1$), we will show how to construct an instance $\I'$ of CRMFC-T such that $\I$ has a feasible placement with dilation 1 iff $\I'$ has a feasible solution. Also, from a feasible solution for $\I'$, a feasible solution for $\I$ can be computed in polynomial time. In the constructed instance $\I'$, the height of the tree is one more than the number of radii classes in NUkC. We show that CRMFC-T can be solved in polynomial time if the height of the input tree is a constant (Appendix \ref{ap:treealgo}). From Lemma \ref{lem:dila1}, it follows that the perturbation resilient version of NUkC can be solved in polynomial time if the number of classes is a constant. Thus, we obtain the following theorem.

\begin{theorem}
NUkC under 2-PR (or $(3,\epsilon)$-PR, where the size of each optimal cluster is more than $\epsilon n+1$) can be solved in polynomial time if the number of radii classes is a constant.  
\end{theorem}


\subsection{Tree Construction}
Let $G$ be the complete graph that defines the distances between the input points. Note that we are also given the input radii $r_1 > r_2 > \ldots >r_t$. We construct the tree $T$ in $t$ rounds that contains $t$ levels other than the root level. We denote the nodes at level $i$ by $L_i$ for $i\in \{0,\ldots,t\}$. $L_0$ contains a singleton node -- the root of the tree. For $i\ge 1$, in $i^{th}$ round, we construct the nodes $L_i$ and connect them with the nodes in $L_{i-1}$. Each node $v$ in $T$ corresponds to a connected subgraph $G_v$ of $G$. The root corresponds to $G$ itself. Also, each node is marked with either yes or no denoting if the node can be selected or it is in the forbidden set.

For each index $i \in \{1,\ldots,t\}$, in $i^{th}$ round, we consider all the nodes $v\in L_{i-1}$ and the subgraph $G_v$ corresponding to $v$. We remove all the edges with weight more than $r_i$ from $G_v$. Let $G_v^1,\ldots,G_v^l$ be the connected components formed from $G_v$ due to the removal of these edges. We add $l$ children of $v$ to $L_i$ corresponding to these connected $l$ subgraphs. For each such child $u$, if there is a node $w$ in $G_u$, such that for all node $x$ in $G_u$, $d(w,x)\le r_i$, we label $u$ with yes. Otherwise, we label $u$ with no (forbidden). Lastly, for each level $i\ge 1$, the number of nodes that can be chosen from $L_i$ in CRMFC-T is set to $k_i$. The following lemma establishes the connection between the two instances $\I$ and $\I'$. 

\begin{lemma}
$\I$ has a feasible placement with dilation 1 iff $\I'$ has a feasible solution to CRMFC-T.  
\end{lemma}

\begin{proof}
 First, suppose there is a feasible solution to $\I'$. For each chosen node $v$, $v$ must be a yes node. Let $i$ be the integer such that $v\in L_i$. Then, the points in $G_v$ can be covered by an $r_i$ ball centered at some point in $G_v$. We choose this ball in our placement. Note that we select at most $k_i$ balls of radius $r_i$ for all $i$. We prove that each point is covered in the constructed placement. Consider any point $p$. The way we construct the tree, each point can lie in the connected subgraph $G_v$ of exactly one node $v$ of $L_j$ for all $j$. Let $\pi$ be the root-leaf path in $T$, such that for any $v\in \pi$, $p$ is in $G_v$. Now, there must be a node along $\pi$ that is chosen in the solution of CRMFC-T. Let $u$ be such a node. As we place a ball of radius $r_i$ that covers all the points of $G_u$, $p$ gets covered. Thus, $\I$ has a feasible placement with dilation 1. 
 
 Now, suppose $\I$ has a feasible placement with dilation 1. Let $\Op$ be the clustering induced by the placement. Now, consider any cluster $C\in \Op$, which is covered by a ball of radius $r_j$. Thus, c-radius$(C,d)\le r_j$. The way the tree $T$ is constructed it follows that all the points in $C$ remain in the same connected subgraph $G_v$ corresponding to a unique vertex $v\in L_i$ for each $i \le j$. Let $G_u$ be the subgraph corresponding to level $j-1$. As $\I$ is a 2-PR (resp. $(3,\epsilon)$-PR) instance, from Lemma \ref{lem:well-separated} (resp. Lemma \ref{lem:well-separatedrobust}), we know that, for any $p\in C$ and $q\in P\setminus C$, $d(p,q)>r_j$. Thus, when the edges with weight more than $r_j$ are removed from $G_u$, $p$ and $q$ cannot remain in the same component. But, as c-radius$(C,d)\le r_j$ all the points of $C$ remain in the same component. Also, by the first property of Observation \ref{obs:uniquecluster}, the center of the $r_j$-ball that covers $C$ must lie in $C$.   It follows that there is a yes node $C(v) \in L_j$ such that $G_{C(v)}$ contains only the points of $C$ as vertices. For each cluster $C\in \Op$, we select the yes node $C(v)$ in the solution to CRMFC-T. It is not hard to see that we choose at most $k_j$ nodes from $L_j$. Now, consider any root-leaf path $\pi$ in $T$ corresponding to a leaf $l$. Let $p$ be a point in $G_l$. Also, let $p$ be a point in the cluster $C\in \Op$. Then, there must be a yes node $C(v)$ in $\pi$ such that $G_{C(v)}$ contains only the points of $C$. As we choose $v$ in our solution, we have at least one node from the path $\pi$. Hence, the constructed solution is feasible. 
\end{proof}

\section{Hardness of Approximation}\label{sec:hardness}
In this section, we will prove the following theorem. 

\begin{theorem}\label{th:nukchardpr}
 For any constant $c$ and any $\gamma\le c^{n^c}$, NUkC under ${\gamma}$-PR is hard to approximate in polynomial time within a factor of ${\gamma}$, unless $\np=\rp$.
\end{theorem}


To prove this theorem, we use a chain of reductions that involves the following problems. \\\\ 
%
%
%
%
%
%
1-in-3SAT \cite{Schaefer:1978}

INSTANCE: An ordered pair $(B,C)$ consisting of a set $B$ of Boolean variables and a set $C$ of clauses over $B$ having three literals each in conjunctive normal form.

QUESTION: Is there a truth assignment for $B$ such that every clause in $C$ contains exactly one true literal?\\\\
\uppercase{Resource Minimization for Fire Containment on Trees (RMFC-T)} \cite{FinbowKMR07,KingM10}

INSTANCE: A rooted tree $T$ and an integer $m$.

QUESTION: Is there a set $N$ of non-root nodes such that every root-leaf path contains a node from $N$ and for any integer $j\ge 1$, $|N\cap L_j|\le m$, where $L_j$ is the set of nodes at distance exactly $j$ from the root?

The chain of reductions that we use consists of the following reductions: (1) 3SAT to 1-in-3SAT, (2) 1-in-3SAT to RMFC-T, and (3) RMFC-T to NUkC. Note that NUkC under PR has a unique optimal solution. As we would like to show hardness for the PR version of NUkC, we will consider ``Unambiguous'' version of all these problems. For ``Unambiguous'' version of 3SAT and 1-in-3SAT, if an instance has a feasible solution, the solution is unique. For ``Unambiguous'' version of RMFC-T, if an instance has a feasible solution, the solution has a specific structure that we will define shortly. For the reduction from 3SAT to 1-in-3SAT, we ensure that the reduction preserves the number of solutions. Such a reduction is called a parsimonious reduction. To refer to the Unambiguous version of a problem we add a prefix `U-' to the problem name. Next, we discuss the details of the reductions. 

In a celebrated work, Valiant and Vazirani \cite{valiant1985np} showed that U-3SAT is hard, unless $\np=\rp$. Schaefer \cite{Schaefer:1978} showed a reduction from 3SAT to 1-in-3SAT to prove the $\np$-hardness of the latter problem. As noted in \cite{abs-1808-02821} the reduction is parsimonious. We use the same reduction (now from U-3SAT to U-1-in-3SAT) to prove the hardness of U-1-in-3SAT, unless $\np=\rp$.

Next, we discuss the reduction from 1-in-3SAT to RMFC-T. First, we define the Unambiguous version of RMFC-T. For a vertex $v$ of a rooted tree $T$, let leaves$(T_v)$ be the set of leaves at the subtree rooted at $v$. For any two feasible solutions $S_1$ and $S_2$ of RMFC-T, $S_1$ and $S_2$ are called equivalent, if the two sets $\cup_{v\in S_1}$ $\{$ leaves$(T_v)\}$ and $\cup_{v\in S_2}$ $\{$ leaves$(T_v)\}$ are identical. U-RMFC-T is same as RMFC-T except if the input instance has more than one feasible solutions, then all the feasible solutions are pairwise equivalent. The reduction from U-1-in-3SAT to U-RMFC-T appears in the appendix. The reduction is a non-trivial adaptation of the reduction due to Finbow {\em et al.}~\cite{FinbowKMR07} from a version of 3SAT (RESTRICTED NAE 3-SAT) to the RMFC-T problem. We summarize our finding in the following lemma. 

\begin{lemma}\label{th:fire}
Given a tree $T$,
it is not possible to distinguish between the following two cases in polynomial time, unless $\np=\rp$.
 \begin{itemize}
  \item \textbf{YES}: There is a solution to the \emph{U-RMFC-T} instance with $m=1$.
  \item \textbf{NO}: There is no solution to the \emph{U-RMFC-T} instance with $m=1$.
 \end{itemize} 
\end{lemma} 

To complete the chain of reductions, now we discuss the last reduction. In particular, we show a reduction from RMFC-T to NUkC that proves the following theorem.

\begin{theorem}\label{th:nukchardtree}
 For any constant $c$ and any $\gamma\le c^{n^c}$, NUkC is $\np$-hard to approximate within a factor of $\gamma$ in tree metrics.
\end{theorem}


Note that this theorem has already been proved in \cite{ChakrabartyGK16}. However, it is not straightforward to show that the instances of NUkC they construct are perturbation resilient. Using a similar construction, we will argue that the instances of NUkC to which the instances of RMFC-T map are perturbation resilient. However, to ensure that the constructed instance of NUkC has a unique optimal solution, we will consider the Unambiguous version of RMFC-T. 

\subsection{Proof of Theorem \ref{th:nukchardtree}}\label{subsec:nukchardtree}
To prove the theorem we show a reduction from U-RMFC-T. As mentioned before, the reduction is similar to the reduction used by Chakrabarty et al. \cite{ChakrabartyGK16}. The construction is as follows. Let $h$ be the height of the tree. We set $P$ to be the leaves of the given tree $T$, i.e., $P=L_h$. For any edge $(u,v)$ of $T$ such that $u\in L_{h}$ and $v\in L_{h-1}$, assign a weight $(\gamma+1)/2$ to $(u,v)$. For any edge $(u,v)$ of $T$ such that $u\in L_{i}$ and $v\in L_{i-1}$ for $i \le h-1$, assign a weight $((\gamma+1)^{h-i+1}-(\gamma+1)^{h-i})/2$ to $(u,v)$. Then the distance function $d$ is the shortest-path metric on $P$ induced by the weights of $T$. We set $t=h$, $r_t=0$ and for any $1\le j< t$, $r_{j}=(\gamma+1)^{t-j}$. Also $k_1=\ldots=k_t=1$. Now we have the following observation.

\begin{observation}[$*$]
\label{obs:ances}
 For any two leaves $u,u'$ with a common ancestor $v\in L_j$, $d(u,u') \le r_j$.
\end{observation}

We note that the weight of any edge is bounded by $(\gamma+1)^{h-1}=c^{O(n^ch)}$ and thus can be represented using $O(n^ch)$ number of bits. It follows that the construction can be done in polynomial-time. We denote the constructed instance of NUkC by $I$. For simplicity, we use the terms point and leaf interchangeably. The following lemma completes the proof of Theorem \ref{th:nukchardtree} which follows from the construction and the fact that the feasible solutions for $T$ are pairwise equivalent. 

\begin{lemma}
[$*$] 
\label{lem:gammahard}
If $T$ is the ``YES'' case of Lemma \ref{th:fire}, then the optimum dilation of $I$ is 1. If $T$ is the ``NO'' case of Lemma \ref{th:fire}, then the optimum dilation of $I$ is more than $\gamma$. Moreover, $I$ has a unique optimal clustering. 
\end{lemma}

\subsection{Hardness of Perturbation Resilient Version of NUkC}
To show the hardness of the $\gamma$-perturbation-resilient version of NUkC, we prove that the constructed instances of U-NUkC in the reduction from U-RMFC-T to U-NUkC in tree metrics are $\gamma$-PR. 
First, we remind the reader of the tree metric $d^*$ we used there. We are given a parameter $\gamma$ and a tree $T_{\gamma}$ with height $h$ whose leaves are at the same distance from the root. The points in the metric space correspond to all the leaves of $T_{\gamma}$. Let $n$ be the number of leaves. Also, let $L_i$ be the nodes of $T_{\gamma}$ at level $i$ for $1\le i\le h$. For an edge $(u,v)$ of $T$ such that $u\in L_{h}$ and $v\in L_{h-1}$, we assign a weight $l(u,v)=(\gamma+1)/2$ to $(u,v)$. For each $u\in L_i$, $v\in L_{i-1}$ for $i \le h-1$ such that $(u,v)$ is an edge in $T_{\gamma}$, we assign a weight $l(u,v)=((\gamma+1)^{h-i+1}-(\gamma+1)^{h-i})/2$. For any two leaves $w,w'$, $d^*(w,w')$ is the length of the shortest path between $w$ and $w'$, i.e., if the least common ancestor of $w,w'$ is in $L_j$, then $d^*(w,w')=(\gamma+1)^{h-j}$. We set $t=h$, $r_t=0$ and for any $1\le j< t$, $r_{j}=(\gamma+1)^{t-j}$. Also, $k_1=\ldots=k_t=1$. Let $L(\gamma)$ be the set of leaves of $T_{\gamma}$. As the distance between any two points and the $r_j$'s are of the form $(\gamma+1)^{i}$ for some $i$, we have the following observation. 

\begin{observation}\label{obs:optdilation}
 The optimal dilation of the instance $\I=\{L(\gamma),d^*,t\}$ is $(\gamma+1)^{i}$ for some integer $i\ge 0$.
\end{observation}

As we have shown before, for any constant $c$ and any $\gamma\le c^{n^c}$, U-NUkC is hard to approximate within a factor of $\gamma$ for the metric space $(T_{\gamma},d^*)$, unless $\np=\rp$. Next, we prove the following lemma.

\begin{lemma}
The instance $\I=\{L(\gamma),d^*,t\}$ is $\gamma$-PR. 
\end{lemma}

\begin{proof}
Let $\Op$ be the optimal clustering of $\I$ and $\alpha$ be its dilation. Consider any $\gamma$-perturbation $d'$ of $d^*$. We prove that the optimal clustering ${\Op}'$ of the instance ${\I}'=\{L(\gamma),d',t\}$ is same as $\Op$. Suppose for the sake of contradiction that ${\Op}'$ is not same as $\Op$. 
As $d'$ is a $\gamma$-perturbation  (the distances are non-increasing), the dilation of ${\Op}'$ is at most $\alpha$. We show that ${\Op}'$ is also a feasible clustering for $\I$ with dilation at most $\alpha$. 
%
%

Consider any non-singleton cluster $C\in {\Op}'$ with center $c_1$ that is covered by an $r_j$-ball for $j < t$. 
Then, 
for all pairs of points $p,q \in C$, $d'(p,q) \le \alpha r_j$. This is true, as all the points are leaves of the tree. From Observation \ref{obs:optdilation}, it follows that $\alpha r_j=(\gamma+1)^{i}$ for some $i$. As $d'$ is a $\gamma$-perturbation of $d^*$, $d^*(p,q) \le \gamma\cdot d'(p,q) < (\gamma+1)^{i+1}$. Now, the way $T_{\gamma}$ is constructed, there is no distance values strictly between $(\gamma+1)^{i}$ and $(\gamma+1)^{i+1}$. Hence, $d^*(p,q) \le (\gamma+1)^{i}=\alpha r_j$, and the ball $B(c_1,\alpha r_j)$ covers the points of the cluster $C$ w.r.t.  $d^*$. 
It follows that ${\Op}'$ is also a feasible clustering for $\I$ with dilation at most $\alpha$. But, as per our assumption $\Op$ and ${\Op}'$ are different, and thus the optimal clustering of $\I$ is not unique. This is a contradiction, and hence $\Op$ and ${\Op}'$ must be same. 
%
%
%
\end{proof}


\subsection{Hardness in Euclidean Metric}

\begin{theorem}\label{th:nukchardeuclid}
 For any constant $\kappa$ and any $\beta\le {\kappa}^{n^{\kappa}}$, NUkC under $\beta$-PR is hard to approximate within a factor of $\beta$ in the Euclidean metric of dimension $d$ for any $d\ge 1$, unless $\np=\rp$.
\end{theorem}

This result is in turn based on the following theorem due to Gupta \cite{Gupta00}.

\begin{theorem}\label{th:embed}
\cite{Gupta00}
 Any weighted tree $T$ with $L$ leaves can be embedded in polynomial-time into $d$-dimensional Euclidean space with $O(dL^{1/(d-1)}\min\{\log L,d\}^{1/2})$ distortion.
\end{theorem}

The idea is to show that if there is a polynomial-time $\beta$-approximation for NUkC under $\beta$-PR in the Euclidean metric for any constant $\kappa$ and any $\beta\le {\kappa}^{n^{\kappa}}$, then there is also a polynomial-time $\gamma$-approximation for NUkC under $\gamma$-PR in tree metrics for any $\gamma\le c^{n^c}$, where $c$ is a constant. But, by Theorem \ref{th:nukchardpr} this is a contradiction, and hence the proof of the theorem follows. To obtain the $\gamma$-approximation in tree metrics we embed the tree metric into Euclidean metric of dimension $d$ using the algorithm of Theorem \ref{th:embed}. Then, we use the algorithm for Euclidean metric to obtain a solution for the embedded instance. Lastly, we map this solution back to the tree metric with sufficient expansion of the balls. For a suitable choice of $\beta$, one can show that the constructed solution is a $\gamma$-approximation. The details are given in the Appendix.

\bibliographystyle{plainurl}
\bibliography{lipics-v2019-sample-article}

\newpage
\appendix

\section{Examples demonstrating the definition of perturbation resilience}\label{ap:examples}

\begin{figure}[t]
\centering
\includegraphics[width=.55\linewidth]{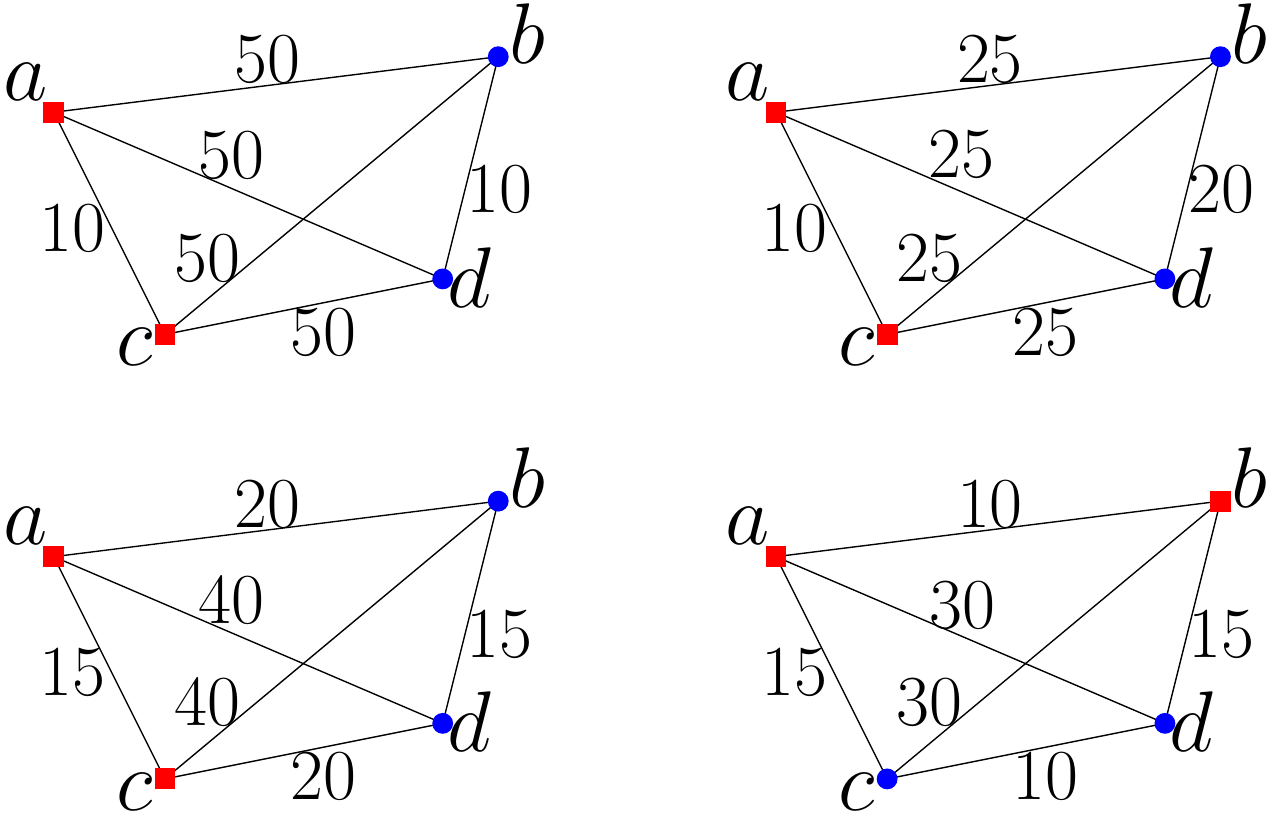}
\caption{Examples demonstrating the definition of perturbation resilience. The top-right (resp. bottom-right) instance is a 2-perturbed instance of the top-left (resp. bottom-left) instance. The points in same optimal cluster are shown by same shape and color.}
\label{fig:pr}
\end{figure}

For more clarity, we describe the notion of $\psi$-perturbation-resilience in the context of NUkC using two examples in Figure \ref{fig:pr} (top-left and bottom-left). In all our examples, the number of clusters $k=2$ and the number of radii classes $t=1$. For the instance shown at the top-left figure, let $r_1=10$. We claim that this instance is $2$-perturbation-resilient. To see this note that here the optimal dilation is 1, and the optimal clusters are $\{a,c\}$ and $\{b,d\}$. Moreover, even if all the distances are perturbed by a factor of 2, the distance between $a$ and $c$ (resp. $b$ and $d$) can be at most 10. Hence, the dilation of the previous clustering for the perturbed instance would be at most $1$. But, as all the distances between $a$ and $b$, $a$ and $d$, $c$ and $b$, and $c$ and $d$ are 50, in any 2-perturbation of the distances, the distance between the two points in any of these four pairs would be at least 25. Thus if both of the points in such a pair remain in same cluster, the dilation must be at least 2.5. As there is a clustering of dilation at most 1, in optimal clustering, both of these points cannot lie in the same cluster. Hence, the optimal clustering is unique and same as the one before. The top-right figure shows a 2-perturbed instance with the same optimal clustering. Now, consider the instance in the bottom-left figure. Let $r_1=15$. We claim that this instance is not 2-perturbation-resilient. To prove this we show a 2-perturbed instance where the optimal clustering is different. Note that in the original instance, the optimal dilation is 1, and the optimal clusters are $\{a,c\}$ and $\{b,d\}$. The 2-perturbed instance we consider is shown in the bottom-right figure. Note that in the perturbed instance the optimal clustering is $\{\{a,b\},\{c,d\}\}$ with dilation 10/15=2/3. This is because any other clustering has a dilation at least 1. 

\section{Proof of Observation \ref{obs:uniquecluster}}
\begin{proof}
\begin{itemize}
 \item Suppose $p$ belongs to the cluster $C'$ such that $C\ne C'$. Construct another clustering ${\Op}'$ by selecting all the clusters in $\Op$ except $C$ and $C'$, and the clusters $C\cup \{p\}$ and $C'\setminus \{p\}$. It is not hard to see that ${\Op}'$ is also a feasible clustering induced by $\Pi$. As $\Pi$ is an optimal placement, ${\Op}'$ is also an optimal clustering, which contradicts the uniqueness of the optimal clustering of $\I$. Hence, the statement follows.  
 
 \item Suppose $B$ contains both $u$ and $v$. We construct a new clustering $\Op'$, which is identical to ${\Op}$ except, in ${\Op}'$, we move the points $u,v$ to the cluster $C$. Note that the clustering ${\Op}'$ can be induced by the placement $\Pi$, as the ball $B$ that covers $C\in {\Op}$ also contains $u,v$. Hence, ${\Op}'$ is an optimal clustering for $\I$ different than ${\Op}$, which is a contradiction, and thus the statement follows.\end{itemize}\end{proof}

\section{Proof of Lemma \ref{lem:dila1}}

\begin{proof}
Consider any instance $\I=(P,d,t)$ of the NUkC problem with $t$ radii classes under $\psi$-PR (resp. $(\psi,\epsilon)$-PR). Let $\alpha$ be the optimal dilation. Note that we do not know the value of $\alpha$. However, as the input metric is finite, there are only polynomial number of guesses for $\alpha$. We use the following procedure to obtain the optimal clustering for $\I$. In each step, we guess a value $\alpha'$ for the optimal dilation in the increasing order of the values. 
We construct a new instance $\I'$ from $\I$ by only changing the radius $r_i$ to $\alpha'\cdot r_i$ for all $i$. 
Then, we apply the algorithm $\A$ on the constructed instance. If $\A$ returns ``no'', we repeat the process with a different guess. Otherwise, the procedure terminates. We return the same clustering returned by $\A$ as the solution for the instance $\I$. 
%
%

Now, we argue about the correctness of the procedure. 
First, we claim that $\I'$ is a $\psi$-PR (resp. $(\psi,\epsilon)$-PR) instance. Before proving this claim we discuss its consequences.   
Note that if there is no feasible solution for $\I$ with dilation $\alpha'$, then with $k_i$ balls of radius $\alpha'\cdot r_i$ for all $i$ it is not possible to cover the input points. Hence, in this case, for the constructed instance, there is no feasible solution with dilation 1. Thus, the algorithm correctly returns ``no'' assuming $\I'$ is a $\psi$-PR (resp. $(\psi,\epsilon)$-PR) instance. 
If there is a feasible solution for $\I$ with dilation $\alpha'$, then with $k_i$ balls of radius $\alpha'\cdot r_i$ for all $i$ one can cover the input points. 
Thus, in that case, for the constructed instance, there is a feasible solution with dilation 1. 
%
Hence, $\A$ correctly returns ``yes'' assuming $\I'$ is $\psi$-PR (resp. $(\psi,\epsilon)$-PR). Thus, when $\alpha'=\alpha$, $\A$ returns ``yes'' and the returned clustering is optimal for $\I$. Now, we prove the claim.     
\begin{claim}
 $\I'$ is a $\psi$-PR (resp. $(\psi,\epsilon)$-PR) instance.
\end{claim}

\begin{proof}
First, we show that the optimal clustering of $\I'$ is unique. Note that the optimal dilation of $\I'$ is $\alpha/\alpha'$. Suppose optimal clustering of $\I'$ is not unique. Then, there are two different clusterings where the points can be covered using $k_i$ balls of radius $(\alpha/\alpha')\cdot \alpha'\cdot r_i=\alpha\cdot r_i$ from each class $i$. It follows that there are two different optimal clusterings for $\I$. But, this is a contradiction, and thus the optimal clustering of $\I'$ is unique. Note that the optimal clusterings of $\I$ and $\I'$ are identical. Let $\C$ be that clustering. Now, consider any $\psi$-perturbation $d_1$ of the input metric $d$ and the $\psi$ perturbed instance $\I_1'$ of $\I'$. Let $\I_1= (P,d_1,t)$ be the corresponding $\psi$ perturbed instance of $\I$. 
Also, let $\C_1'$ be the optimal clustering of $\I_1'$ with dilation $\alpha_1'$. For the sake of contradiction, suppose $\C_1'$ is not identical (resp. $\epsilon$-close) to $\C$. We argue that $\C_1'$ is also an optimal clustering of $\I_1$. But, this is a contradiction, as $\I_1$ is a $\psi$ perturbed instance of $\I$ and $\I$ is a $\psi$-PR (resp. $(\psi,\epsilon)$-PR) instance. Now, note that a placement that induces the clustering $\C_1'$ of $\I_1'$ uses $k_i$ balls of radius $\alpha_1'\cdot \alpha'\cdot r_i$ from each class $i$. Thus, $\C_1'$ is a clustering for $\I_1$ with dilation $\alpha_1'\cdot \alpha'$. It is sufficient to argue that this dilation is optimal for $\I_1$. Suppose the optimal dilation is $< \alpha_1'\cdot \alpha'$. Then, using $k_i$ balls of radius $< \alpha_1'\cdot \alpha'\cdot r_i$ from each class $i$ all the points can be covered. Hence, there is a clustering for $\I_1'$ with dilation $< \alpha_1'$, which is a contradiction, and hence the claim follows.  
\end{proof}
Finally, as the number of guesses for $\alpha$ is a polynomial, the procedure terminates in polynomial time.  
\end{proof}

\section{Proof of Lemma \ref{lem:well-separated}}

\begin{proof}
 Let $\Op$ be the optimal clustering induced by the placement $\Pi$. Also, let $B_1$ and $B_2$ be the balls that induce the clusters $C_1$ and $C_2$, respectively. For the sake of contradiction, suppose there exist two points $p\in C_1$, $q\in C_2$ such that $d(p,q) \le r_i$. The idea is to show that there is a metric $d_1$ that is a $2$-perturbation of $d$ such that $\I'=(P,d_1,t)$ has different optimal clustering than $\Op$. But, this is a contradiction, and thus the lemma follows. 
 
 Let $c_t$ be the center of the ball $B_t$ for $t\in \{1,2\}$. Then, $d(c_1,q)\le d(c_1,p)+d(p,q)\le 2r_i$. We define the distance function $d_1$ in the following way. First, we construct the complete graph with vertex set equal to $P$, and for any edge $(u,v)$, its length is defined by the function $l$.
%
 \[l(u,v) =
\left\{
	\begin{array}{ll}
		\min\{d(u,v),r_i\}  & \mbox{if } u=c_1 \text{ and } v=q\\
		d(u,v) & \mbox{otherwise} 
	\end{array}
\right.\]
We note that, for any $u,v$, $d(u,v)/2\le l(u,v)\le d(u,v)$. The distance function $d_1$ is defined by the shortest path distance between any pair of vertices. It is not hard to verify the following observation.

\begin{observation}
 $d_1$ is a metric $2$-perturbation of $d$. 
\end{observation}

Hence, the instance $\I'=(P,d_1,t)$ has the same optimal clustering $\Op$. Next, we prove a claim that the optimal dilation of $\I'$ is also 1. 

\begin{claim}
 The optimal dilation of $\I'$ is 1.
\end{claim}

\begin{proof}
  As for any $u,v \in V$, $d_1(u,v)\le d(u,v)$, the optimal dilation of $\I'$ is at most 1. We prove that this dilation is at least 1. Suppose the dilation is less than 1. 
 Let $\Pi'$ be any placement with dilation less than 1 that induces the clustering ${\Op}$ of $\I'$. Then, we show that there is a placement for $\I$ with dilation less than 1. But, this is a contradiction, and hence the claim follows. Consider any cluster $C\in \Op$ that gets covered by an $r_t$-ball $B=B(w,r)$ in $\Pi'$. Let $x$ be any point in $C$. Now, consider the distance $d_1$. Let $\pi$ be any shortest path between $w$ and $x$. We claim that $\pi$ cannot contain the edge $(c_1,q)$. For the sake of contradiction, say $\pi$ contains $(c_1,q)$. But, this implies $d_1(w,c_1)\le d_1(w,x)\le r$ and $d_1(w,q)\le d_1(w,x)\le r$. Thus, $B$ contains both $c_1$ and $q$. Now, by the first property of Observation \ref{obs:uniquecluster}, $c_1$ belongs to $C_1$. Thus, by the second property of Observation \ref{obs:uniquecluster}, we obtain a contradiction, as $q\in C_2$. Hence, $\pi$ does not contain $(c_1,q)$. It follows that $d_1(w,x)\ge d(w,x)$. Thus, the radius of the ball needed to cover the points of $C$ in $\I$ is at most $r$. Hence, it is sufficient to use an $r_t$-ball with at most $r/r_t < 1$ factor expansion to cover the points of $C$ in $\I$. Now, we construct a placement for $\I$ by selecting the same balls to cover the clusters that are used in $\Pi'$. Clearly, the dilation of this placement is less than 1. 
\end{proof}

Next, we show that there is a different clustering ${\Op}'$ of $\I'$ with exactly $k$ clusters that achieves the optimal dilation. This gives rise to a contradiction, and thus $d(p,q) > r_i$. Now, there are two cases. In the first case, 
$q$ is the only point in $C_2$, and thus $C_2\setminus \{q\}$ is empty. In this case, we pick a non-singleton cluster $C$ from ${\Op}\setminus \{C_1\}$ and choose a point $s\in C$. 
Such a cluster exists WLOG. Then, we define ${\Op}'$ to be the set of clusters in $\Op$ except $C,C_1$ and $C_2$, and the clusters $C_1\cup \{q\},\{s\}$ and $C\setminus \{s\}$. In the second case, $q$ is not the only point in $C_2$, and thus $C_2\setminus \{q\}$ is not empty. In this case, ${\Op}'$ is defined to be the set of clusters in $\Op$ except $C_1$ and $C_2$, and the clusters $C_1\cup \{q\},C_2\setminus \{q\}$. It is not hard to see that $C_1\cup \{q\}$ can be covered by the ball $B(c_1,r_i)$. Also, if $C_2\setminus \{q\}$ is not empty, then $B(c_2,r_j)$ covers the points in $C_2\setminus \{q\}$. Hence, in all the cases, it is trivial to verify that the dilation of the new clustering is 1. 
\end{proof}

Note that, in the above proof, to show that ${\Op}'$ has dilation 1, we argue that there is a placement with dilation 1. The balls in the placement might not be disjoint (both $B(c_1,r_i)$ and $B(c_2,r_j)$ cover $q$). But, for the sake of just showing the optimality of the clustering, it is sufficient to show the existence of such a placement.

\section{The Algorithm for CRMFC-T} \label{ap:treealgo}
In this section, we design a dynamic programming based algorithm that decides the feasibility of any instance of CRMFC-T. The algorithm runs in polynomial time when the height of the tree is a constant. Let $T$ be the input tree having height $t$, i.e., $T$ has $t+1$ levels $L_0,\ldots,L_t$. $L_0$ contains only the root of $T$. Let $n_i=|L_i|$. We also assume that the nodes of $L_i$ are ordered for all $i\ge 1$, i.e., $L_i=\{v_{i1},\ldots,v_{in_i}\}$. For $j\le l$, let $F(i,j,l)$ be the union of the induced subtrees of $T$ rooted at the vertices $v_{ij},\ldots,v_{il}$. We construct the tree $T(i,j,l)$ from $F(i,j,l)$ by connecting the roots of the subtrees to a common root. 

Let feasible($T(i,j,l),l_i,l_{i+1},\ldots,l_t$) be the function that decides if there is a feasible solution to CRMFC-T for the tree $T(i,j,l)$ by selecting at most $l_m$ nodes from level $m$, where $i\le m\le t$. Note that computing the function feasible($T=T(1,1,n_1),k_1,\ldots,k_t$) solves the CRMFC-T problem. We consider the following recursive definition of feasible(). In the base case, if $i=t-1$, the function can be computed in polynomial time. Otherwise, if $l_i$ is 0, let $j'$ be the minimum index such that $v_{i+1,j'}$ is a child of $v_{ij}$ and $l'$ be the maximum index such that $v_{i+1,l'}$ is a child of $v_{il}$. In this case, feasible($T(i,j,l),l_i,l_{i+1},\ldots,l_t$)=feasible($T(i+1,j',l'),l_{i+1},\ldots,l_t$). Otherwise, there must be a minimum index $j\le j^1\le l$ such that a yes node $v_{ij^1}$ is selected to be in the solution. For such a fixed $j<j^1<l$, let $j'$ be the minimum index such that $v_{i+1,j'}$ is a child of $v_{ij}$ and $l'$ be the maximum index such that $v_{i+1,l'}$ is a child of $v_{i,j^1-1}$. In this case, if there are values $l_{i+1}^1,\ldots,l_t^1,l_i^2,l_{i+1}^2,\ldots,l_t^2$ such that $l_i^2=l_i-1$, $l_m=l_m^1+l_m^2$ for all $i+1\le m\le t$, and both feasible($T(i+1,j',l'),l_{i+1}^1,\ldots,l_t^1$) and feasible($T(i,j^1+1,l),l_i^2,l_{i+1}^2,\ldots,l_t^2$) return yes, then feasible($T(i,j,l),l_i,l_{i+1},\ldots,l_t$) also returns yes. Otherwise if for all $j^1$ there are no such values, feasible($T(i,j,l),l_i,l_{i+1},\ldots,l_t$) returns no. The corner cases when $j^1=j$ or $j^1=l$ can be handled similarly. 

It is not hard to verify that feasible($T(i,j,l),l_i,l_{i+1},\ldots,l_t$) correctly decides whether there is a feasible solution or not for $T(i,j,l)$. To compute the feasible() function for all possible values one can use a simple dynamic programming based technique. In particular, one can store the values of the function for all possible parameters in a table. The table is filled up in a bottom-up manner, where the values corresponding to a level $j$ subtree is computed before computations of the values corresponding to a level $i$ subtree for $i < j$. It is not hard to see that the procedure would take polynomial time and space for a constant $t$.

\section{Reduction from 1-in-3SAT to RMFC-T}
Finbow {\em et al.}~\cite{FinbowKMR07} showed a reduction from Restricted NAE 3-SAT to RMFC-T. As per the definition of Restricted NAE 3-SAT, if the input instance has a feasible assignment, then it must at least have two. Thus, it cannot have a unique feasible solution. This is the reason behind our selection of the problem 1-in-3SAT, which can have a unique feasible solution. However, the reduction is motivated by the one in \cite{FinbowKMR07}. For consistency, we borrow some of their notations. 

Given an instance $I$ of 1-in-3SAT, we construct a rooted tree $T$ with root $r$ in multiple steps. Also, we choose the parameter $m=1$. Before discussing the reduction, we have a few definitions to set up the stage. Throughout this discussion, we will use the operation \emph{root} a copy of a rooted tree $(T,r)$ at a vertex $x$ of a graph $G$. This means we construct a new graph from the disjoint union of $G$ and $T$ by identifying $x$ and $r$. A vertex $v$ of a tree is said to be \emph{defended} by a vertex $u$ if the root to $v$ path contains $u$. For any path, we assume that its root is one of the degree one vertices. Also, the length of a path is defined as the number of edges contained in it. 

A \emph{ladder tree} ${\cl}^{T}(n)$ is a path having $2n+1$ vertices such that the middle vertex of the path is identified as the root of the tree. See Figure \ref{fig:trees}(i). Thus, the root of ${\cl}^{T}(n)$ has two branches each being a path of length $n$. A \emph{bell tree} ${\cb}^{T}(n,m)$ is formed by rooting a ladder tree ${\cl}^{T}(n-m)$ at an endpoint of a path having $m$ edges. The other endpoint of the path becomes the root of the bell tree. See Figure \ref{fig:trees}(ii). Thus, in the figure, the distance (in terms of edges) between $a$ and $b$ is $m$ and the distance between $a$ and a leaf is $n$. A \emph{snake tree} ${\cs}^{T}(n,m)$ is formed by rooting an $m-1$ length path at the root of a bell tree ${\cb}^{T}(n,m+1)$. The root of the bell tree (or the path) becomes the root of the snake tree. Note that a snake tree has exactly one degree 3 vertex. See Figure \ref{fig:trees}(iii). Thus, in the figure, the distance between $a$ and $b$ is $m$, and the length of the path between $a$ and a leaf such that the path contains $b$ is $n$. A rooted tree $T$ is called full if all leaves occur at the same level. A rooted tree $T$ is called complete if every internal vertex has exactly two children. One simple observation is that a complete and full binary tree of height $h\ge 0$ has $2^{h+1}-1$ vertices, and among those $2^h$ are leaves. 

\begin{figure}[t]
\centering
\includegraphics[width=.5\linewidth]{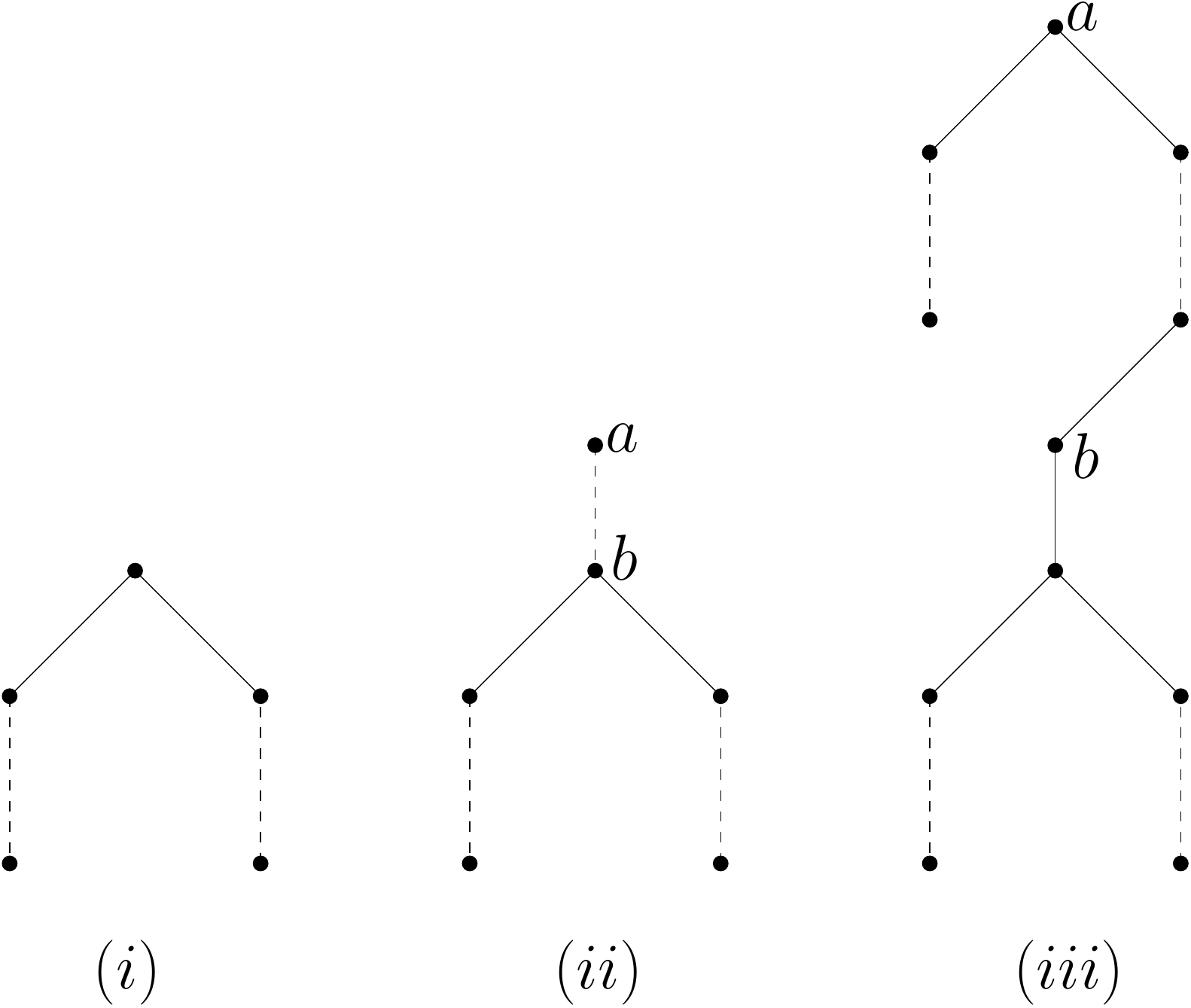}
\caption{(i) A ladder tree. (ii) A bell tree. (iii) A snake tree. Dashed segments denote paths.}
\label{fig:trees}
\end{figure}

Now, we describe the construction. We are given the 1-in-3SAT instance $I$=$(B,C)$ with the set of variables $B=\{b_1,\ldots,b_b\}$ and the set of clauses $C=\{C_1,\ldots,C_n\}$. Let $p=\lceil{\log n}\rceil+2$. Thus, $2^p\ge 4n$. We are going to construct a tree $T$ which is initialized to the root vertex $r$. For each $1\le i\le b$, root two paths of length $i$ at the root $r$ of $T$. Call the degree one vertices of these two paths $b_i$ and $\overline{b_i}$. Root a complete and full binary tree of height $p$ at $b_i$ and $\overline{b_i}$ for each $i$. From each leaf of these trees root a path of length $b-i$. Call the leaves of these paths $t_{b_i,1},\ldots ,t_{b_i,2^p}$ and $t_{\overline{b_i},1},\ldots ,t_{\overline{b_i},2^p}$. Note that all the leaf nodes are now at a distance $b+p$ from $r$. Root two paths of length $b+1$ at $r$, and call the degree one vertices of these paths $b_0$ and $\overline{b_0}$. So far the construction is exactly the same as the one in \cite{FinbowKMR07}. In the following, we modify their construction to adapt it for our setting. From $b_0$ and $\overline{b_0}$ root a complete and full binary tree of height $p$ and $p+1$, respectively, and call their leaves $t_{b_0,1},\ldots ,t_{b_0,2^{p}}$ and $t_{\overline{b_0},1},\ldots ,t_{\overline{b_0},2^{p+1}}$. This completes the first phase of the construction (see Figure \ref{fig:construction}). 

\begin{figure}[t]
\centering
\includegraphics[width=.6\linewidth]{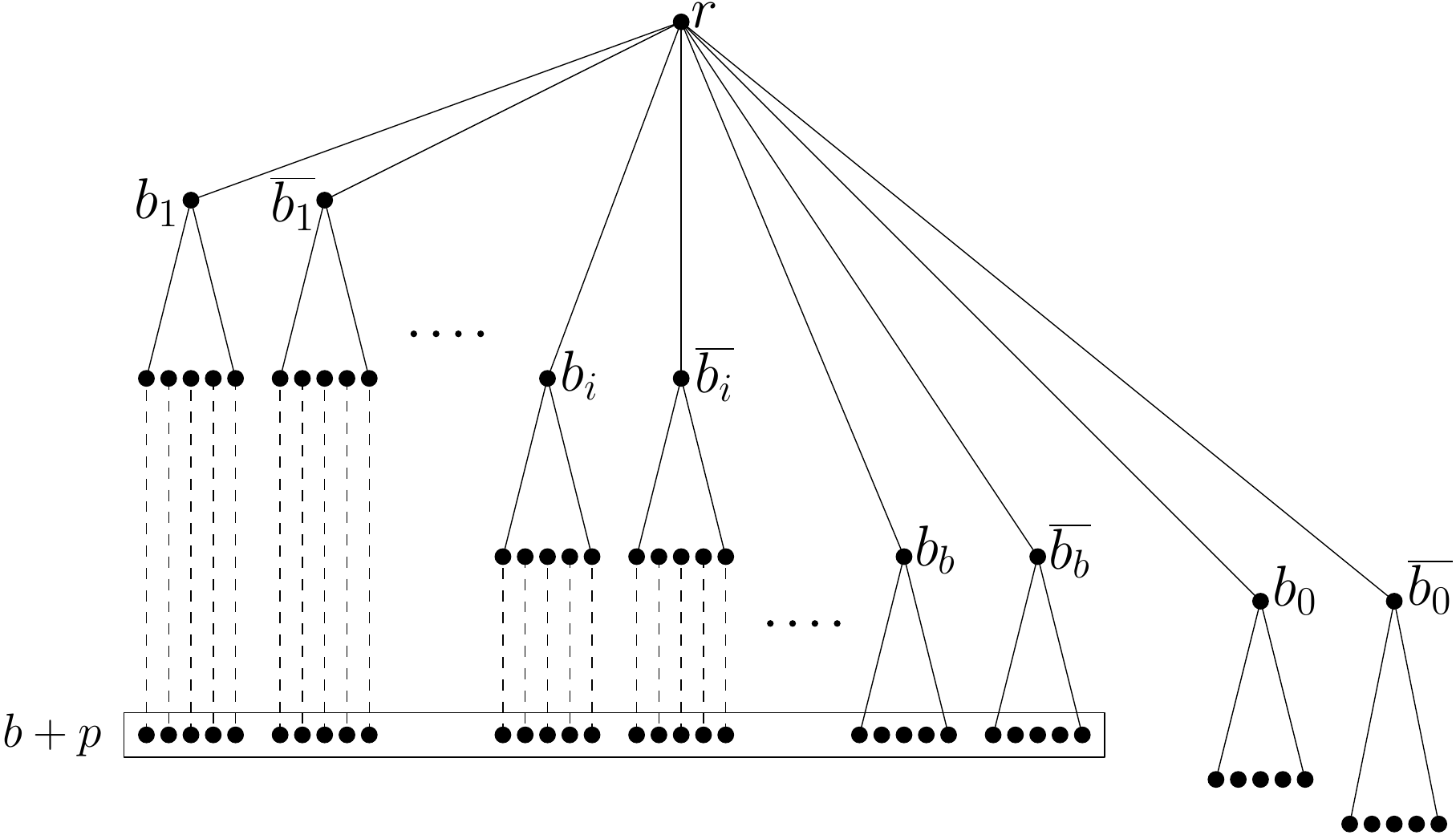}
\caption{Figure showing the constructed tree after the first phase.}
\label{fig:construction}
\end{figure}

In the second phase, we add clause gadgets by rooting special tree structures at the leaves of $T$ constructed so far. For each $1\le j\le n$, and for each literal $l$ of $C_j$, root the snake tree ${\cs}^T(4n+3,4j-2)$ at $t_{\overline{l},j}$. For $1\le \tau\le 3$, let $q$ be the ${\tau}^{th}$ literal of $C_j$. Root ${\cs}^T(4n+3,4j-2+\tau)$ at $t_{q,j}$. Also, root the bell tree ${\cb}^T(4n+3,4j+1)$ at $r$ for each $1\le j\le n$. For $1\le i\le 2^{p}$, add two children $x_i$ and $y_i$ of $t_{b_0,i}$. In each such added child, root the ladder tree  ${\cl}^T(4n+1)$. For $1\le i\le 2^{p+1}$, add two children $\overline{x_i}$ and $\overline{y_i}$ of $t_{\overline{b_0},i}$. In each such added child, root the ladder ${\cl}^T(4n)$. At each remaining leaf of $T$ (as mentioned before) where no tree structure has been rooted so far, root ${\cl}^T(4n+3)$. This completes the construction.     

Now, let us give an intuitive description of the clause gadgets. Note that our main goal is to defend all the leaves. Consider the clause $C_j=(b_1 \vee b_2\vee b_3)$. In a feasible solution, exactly one literal of $C_j$ must be true, say $b_1$. Now suppose in the solution of U-RMFC-T we select the vertices corresponding to true literals, i.e., $b_1$, $\overline{b_2}$ and $\overline{b_3}$. Note that we have added one snake tree corresponding to each complemented literal of $C_j$. Thus, all the vertices in the snake trees corresponding to $\overline{b_2}$ and $\overline{b_3}$ are already defended. In this case, we can defend the degree three vertex (and all of its descendants) of the snake tree corresponding to $\overline{b_1}$ by choosing the degree three vertex itself. If more than one literal are true, then we need to defend vertices of at least two snake trees instead for which we would have to pick more than one vertices from a level. Now, we have also added three other snake trees one for each literal of $C_j$. As the snake tree corresponding to $b_1$ is already defended by $b_1$, we just need to defend the leaves of the remaining two. We can defend them by selecting the parent of the degree three vertex from the corresponding snake tree. In this way, we can also defend the last added bell tree by selecting its degree three vertex  (see Figure \ref{fig:clauseg}). The alignments of these degree three vertices and their parents help us pick them in different levels. Note that if none of the literals are true, then we would need to defend the leaves of the three snake trees corresponding to the literals and in that case it is not possible to defend the leaves of the bell tree corresponding to $C_j$. 

\begin{figure}[t]
\centering
\includegraphics[width=.6\linewidth]{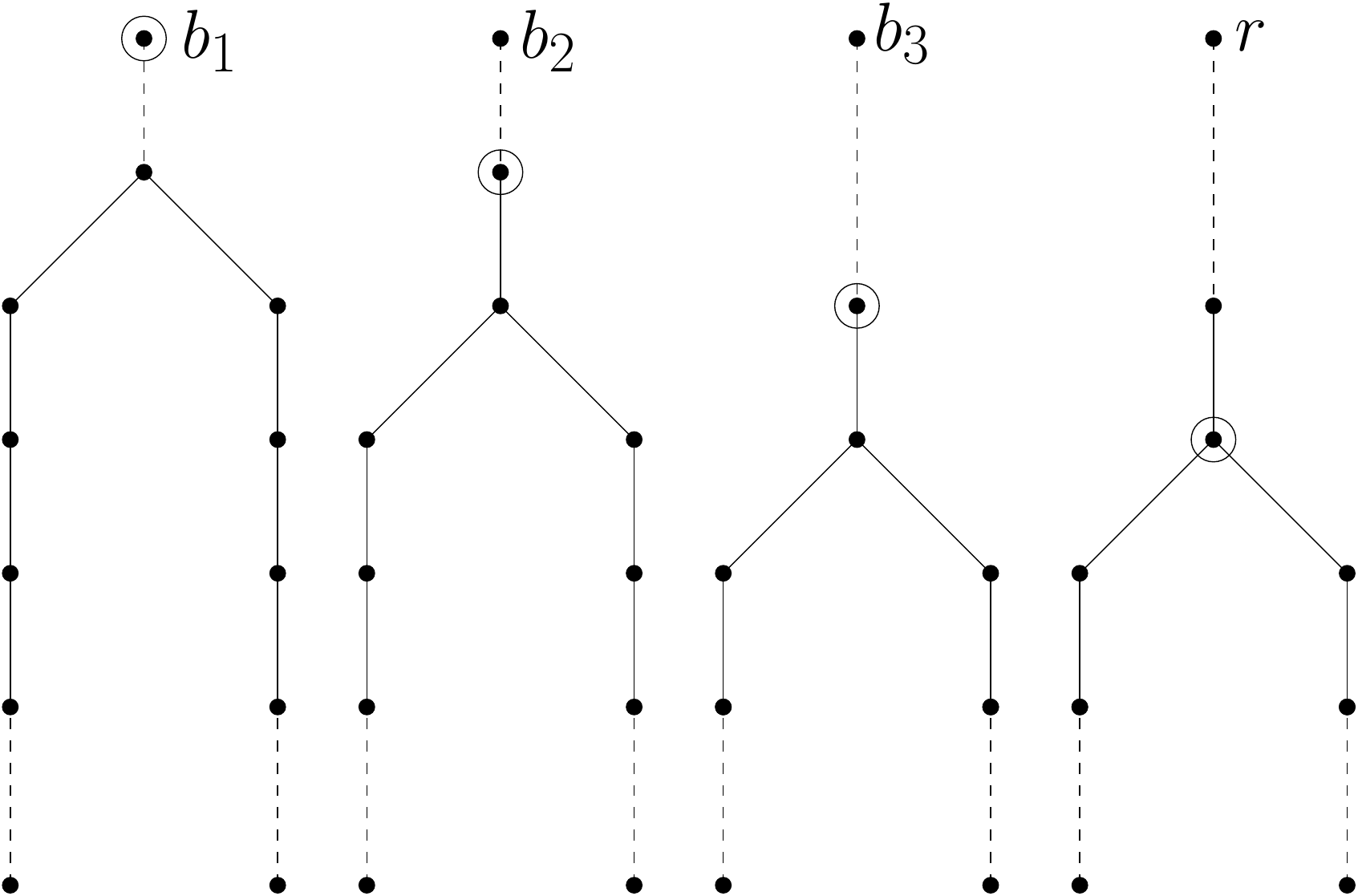}
\caption{Figure showing parts of the three snake trees and the bell tree corresponding to the literals of $C_j$. The circled vertices are selected in the solution.}
\label{fig:clauseg}
\end{figure}

The argument behind the correctness of the reduction is similar to the one in \cite{FinbowKMR07}. The forward direction is simple. First, defend the vertices corresponding to true literals, i.e., if $b_i$ (resp. $\overline{b_i}$) is true, defend $b_i$ (resp. $\overline{b_i}$) at time $i$ for $1\le i\le b$. At time $b+1$, defend $\overline{b_0}$. From time $b+2$ to $b+p+1$, defend the unprotected descendant of $b_0$ which is not on the path from $r$ to $x_1$. At time $b+p+2$, defend $x_1$. From time $b+p+3$ to $b+p+4n+3$, defend the tree greedily by picking a vertex at each level that contains the maximum number of nodes in the subtree rooted at it. The other direction is nontrivial, but similar counting arguments as in \cite{FinbowKMR07} should be used for the proof. It follows that the 1-in-3SAT formula is satisfiable if and only if all the leaves of $T$ can be defended by selecting exactly one vertex from each level.    

Finally, we show that all the feasible solutions are pairwise equivalent as claimed. This actually follows from the construction. Fix the unique feasible assignment to the 1-in-3SAT formula. Then, while finding a feasible solution for U-RMFC-T from the assignment in the above, in all the steps one need to select a unique vertex except when one needs to choose the parent of the degree three vertices of ${\cs}^T(4n+3,4j+1)$ and ${\cb}^T(4n+3,4j+1)$ both of which lie at the same level. However, irrespective of the selection, the set of leaves remains same. Thus, even though the solutions are different, the corresponding sets of leaves are same, and hence the solutions are pairwise equivalent. 

\section{Proof of Observation \ref{obs:ances}}

\begin{proof}
\begin{align*}
 d(u,u') & \le d(u,v)+d(v,u')
 \\& = ((\gamma+1)/2+((\gamma+1)^{2}-(\gamma+1))/2+\ldots+((\gamma+1)^{h-j}-(\gamma+1)^{h-j-1})/2)+
 \\& \qquad((\gamma+1)/2+((\gamma+1)^{2}-(\gamma+1))/2+\ldots+((\gamma+1)^{h-j}-(\gamma+1)^{h-j-1})/2)
 \\& =r_j.
\end{align*}

\end{proof}

\section{Proof of Lemma \ref{lem:gammahard}}
\begin{proof}
 Let $T$ be a ``YES'' instance and $N$ be a solution for $T$. We construct a solution for $I$ from $N$ as follows. For any $v\in N$, let $j$ be the integer such that $v\in L_j$. We select a leaf $u$ from the subtree rooted at $v$ and place a ball of radius $r_j$. We note that at most 1 ball of radius $r_i$ is selected for all $i$, as $|N\cap L_i|\le 1$. Now consider any point $w\in P$. Then there must be a node $v$ in $N$ along the path between $w$ and the root. Let $v\in L_j$. Now the way we place the balls there must be a leaf $u$ in the subtree rooted at $v$ such that a ball of radius $r_j$ is opened at $u$. As $v$ is a common ancestor of $u$ and $w$, from Observation \ref{obs:ances}, it follows that $d(u,w)\le r_j$. Hence the ball $B(u,r_j)$ covers $w$.  
 
 Now let $T$ be a ``NO'' instance and the optimum dilation of $I$ be at most $\gamma$. Consider such a solution $S$ corresponding to the instance $I$. We construct a solution $N$ for U-RMFC-T on $T$ using $S$ as follows. For any $1\le j\le t$, let $u$ be the point where the ball (of radius at most $\gamma r_j$) corresponding to $r_j$ is placed. Let $v$ be the ancestor of $u$ that is in $L_j$. We add $v$ to $N$. Note that, as $S$ contains only one ball corresponding to the value $r_i$, $|N\cap L_i|\le 1$ for all $i$. Now consider any leaf $w$. We show that $N$ contains a node along the $w$-root path. Let $B$ be a ball in $S$ that covers $w$. Also let $B$ be corresponding to the value $r_j$ and is centered at the point $u$. Suppose $v$ is the ancestor of $u$ that is in $L_j$. As the radius of the ball at $u$ is at most $\gamma r_j< r_{j-1}$, a point that is not contained in the subtree rooted at $v$ cannot be covered by $B$. Hence $w$ must be contained in the subtree rooted at $v$ and thus $w$-root path contains $v\in N$. But this implies that $N$ is a solution for $T$ corresponding to the ``YES'' case and thus $T$ must be a ``YES'' instance. But this is a contradiction and thus the optimum dilation of $I$ must be more than $\gamma$. 
 
 As the feasible solutions for $T$ are pairwise equivalent, it follows due to argument above that these feasible solutions get mapped to a unique optimal clustering of dilation 1. Similarly, the unique optimal clustering of dilation 1 gets mapped to a feasible solution of $T$. It follows that $I$ has a unique optimal clustering.  
\end{proof}

\section{Hardness in Euclidean Metric}
Let $X$ and $Y$ be two finite metric spaces with metrics $d$ and $d'$, respectively. Let $f:X\rightarrow Y$ be a map. Then, the \textit{contraction} of $f$ is defined as, \[D_c(f)=\max_{x,y\in X} \frac{d(x,y)}{d'(f(x),f(y))}.\] The \textit{expansion} of $f$ is similarly defined as, \[D_e(f)=\max_{x,y\in X} \frac{d'(f(x),f(y))}{d(x,y)}.\] The \textit{distortion} of $f$, $D(f)=D_c(f)\cdot D_e(f)$. We need Theorem \ref{th:embed} due to Gupta \cite{Gupta00} for proving the hardness result. Next, we prove Theorem \ref{th:nukchardeuclid}. 

\begin{proof}
 Suppose there is a polynomial-time $\beta$-approximation for NUkC under $\beta$-PR in the Euclidean metric for any constant $\kappa$ and any $\beta\le {\kappa}^{n^{\kappa}}$. Then, we show that there is a polynomial-time $\gamma$-approximation for NUkC under $\gamma$-PR in tree metrics for any $\gamma\le c^{n^c}$, where $c$ is a constant. But, by Theorem \ref{th:nukchardpr} this is a contradiction, and hence the proof of the theorem follows. 
 
 Now, consider a constant $c$ and any $\gamma \le c^{n^c}$. Also, consider any instance of NUkC under $\gamma$-PR in the tree metric induced by the weighted tree $T$. We show how to get a $\gamma$-approximate solution for $T$ using the approximation algorithm for the Euclidean metric. Let $\Delta=O(dn^{1/(d-1)}\log n)$.
 First, we embed the tree $T$ into $d$-dimensional Euclidean space $\mathbb{R}^d$ using the algorithm of Theorem \ref{th:embed}. Let $f: T \rightarrow \mathbb{R}^d$ be the embedding. Also, let $d$ and $d_f$ denote the tree and the Euclidean metric, respectively. 
 We fix $\beta$ such that $\beta \le \gamma/\Delta$,  
 and compute a $\beta$-approximate solution $S$ of NUkC under $\beta$-PR for the Euclidean instance. Thereafter, we construct a solution $S'$ for the problem on $T$ from the solution $S$ in the following way. For any node $x$ of $T$, if $S$ contains a ball centered at $f(x)$ with radius $r$, then we add the ball at $x$ of radius $D_c(f)\cdot r$ to $S'$, where $D_c(f)$ is the contraction of $f$. First, we show that the solution $S'$ constructed in this way covers all the nodes of $T$. Consider any node $x$ of $T$. Then, there is a ball in $S$ centered at some point $f(y)$ that covers $f(x)$. Let $r$ be the radius of this ball. It follows that $S'$ contains the ball $B$ centered at $y$ having radius $D_c(f)\cdot r$. Now, \[d(x,y)\le D_c(f)\cdot  d_f(f(x),f(y))\le D_c(f) \cdot r.\] Hence, the ball $B$ contains $x$, and thus $S'$ is a feasible solution. Now, we show that the dilation $\alpha(S')$ of the balls in $S'$ is at most $\gamma$ times the optimum dilation. To this end, let $OPT$ and $OPT_f$ be the optimum dilation for the tree and the Euclidean instance, respectively. Then, the dilation $\alpha(S')$ is at most $\beta \cdot OPT_f\cdot D_c(f)$. Now, as the distances between the points can get expanded by a factor of at most $D_e(f)$ due to the embedding, $OPT_f \le D_e(f)\cdot OPT$. Here $D_e(f)$ is the expansion of $f$. Hence, \[\alpha(S') \le \beta \cdot D_e(f)\cdot OPT\cdot D_c(f)= \beta\cdot D(f) \cdot OPT\le \beta\cdot \Delta\cdot OPT\le \gamma \cdot OPT.\] This completes the proof of the theorem. 
\end{proof}

\end{document}